\documentclass[english,showpacs,preprint,superscriptaddress]{revtex4-1}
\usepackage{lmodern}
\usepackage{lmodern}
\usepackage[T1]{fontenc}
\usepackage{amsmath}
\usepackage{amssymb}
\usepackage{esint}

\makeatletter

\usepackage{amsthm}\usepackage{latexsym}\usepackage{bm}\usepackage{amsfonts}

\usepackage{babel}

\usepackage{babel}

\usepackage{hyperref}

\makeatother

\usepackage{babel}
\begin{document}

\title{Local Energy Conservation Law for Spatially-Discretized Hamiltonian
Vlasov-Maxwell System}

\author{Jianyuan Xiao}

\affiliation{School of Nuclear Science and Technology and Department of Modern
Physics, University of Science and Technology of China, Hefei, Anhui
230026, China}

\affiliation{Key Laboratory of Geospace Environment, CAS, Hefei, Anhui 230026,
China}

\author{Hong Qin }
\email{corresponding author: hongqin@ustc.edu.cn}

\selectlanguage{english}%

\affiliation{School of Nuclear Science and Technology and Department of Modern
Physics, University of Science and Technology of China, Hefei, Anhui
230026, China}

\affiliation{Plasma Physics Laboratory, Princeton University, Princeton, NJ 08543,
U.S.}

\author{Jian Liu}

\affiliation{School of Nuclear Science and Technology and Department of Modern
Physics, University of Science and Technology of China, Hefei, Anhui
230026, China}

\affiliation{Key Laboratory of Geospace Environment, CAS, Hefei, Anhui 230026,
China}

\author{Ruili Zhang}

\affiliation{School of Nuclear Science and Technology and Department of Modern
Physics, University of Science and Technology of China, Hefei, Anhui
230026, China}

\affiliation{Key Laboratory of Geospace Environment, CAS, Hefei, Anhui 230026,
China}
\begin{abstract}
Structure-preserving geometric algorithm for the Vlasov-Maxwell (VM)
equations is currently an active research topic. We show that spatially-discretized
Hamiltonian systems for the VM equations admit a local energy conservation
law in space-time. This is accomplished by proving that for a general
spatially-discretized system, a global conservation law always implies
a discrete local conservation law in space-time when the algorithm
is local. This general result demonstrates that Hamiltonian discretizations
can preserve local conservation laws, in addition to the symplectic
structure, both of which are the intrinsic physical properties of
infinite dimensional Hamiltonian systems in physics. 
\end{abstract}



\maketitle
\global\long\def\EXP{\times10}
 \global\long\def\rmd{\mathrm{d}}
 \global\long\def\xs{ \mathbf{x}_{s}}
 \global\long\def\dotxs{\dot{\mathbf{x}}_{s}}
 \global\long\def\bfx{\mathbf{x}}
 \global\long\def\bfv{\mathbf{v}}
 \global\long\def\bfA{\mathbf{A}}
 \global\long\def\bfB{\mathbf{B}}
 \global\long\def\bfS{\mathbf{S}}
 \global\long\def\bfG{\mathbf{G}}
 \global\long\def\bfE{\mathbf{E}}
 \global\long\def\bfu{\mathbf{u}}
 \global\long\def\bfe{\mathbf{e}}
 \global\long\def\bfd{\mathbf{d}}
 \global\long\def\rme{\mathrm{e}}
 \global\long\def\rmi{\mathrm{i}}
 \global\long\def\rmq{\mathrm{q}}
 \global\long\def\ope{\omega_{pe}}
 \global\long\def\oce{\omega_{ce}}
 \global\long\def\FIG#1{Fig.~#1}
 \global\long\def\EQ#1{Eq.~(\ref{#1})}  \global\long\def\SEC#1{Sec.~#1}
 \global\long\def\REF#1{Ref.~\cite{#1}}
 \global\long\def\CURLD{ {\mathrm{curl_{d}}}}
 \global\long\def\DIVD{ {\mathrm{div_{d}}}}
 \global\long\def\CURLDP{ {\mathrm{curl_{d}}^{T}}}
 \global\long\def\cpt{\captionsetup{justification=raggedright }}
 \global\long\def\act{\mathcal{A}}
 \global\long\def\calL{\mathcal{L}}
 \global\long\def\calJ{\mathcal{J}}
 \global\long\def\WZERO#1{W_{\sigma_{0} I}\left( #1 \right)}
 \global\long\def\WONE#1{W_{\sigma_{1} J}\left( #1 \right)}
\global\long\def\WONEA#1{W_{\sigma_{1} I}\left( #1 \right)}
 \global\long\def\WONEJp#1{W_{\sigma_{1} J'}\left( #1 \right)}
 \global\long\def\WTWO#1{W_{\sigma_{2} K}\left( #1 \right)}


The dynamics of a collection of charged particles and electromagnetic
fields is governed by the well-known Vlasov-Maxwell (VM) equations
\begin{alignat}{1}
\frac{\partial f_{s}}{\partial t}+\bfv\cdot\nabla f_{s}+\frac{q_{s}}{m_{s}}\left(\bfE+\bfv\times\bfB\right)\cdot\frac{\partial f_{s}}{\partial\bfv}=0\thinspace, & \text{}\label{eq:VM}\\
\frac{\partial\bfB}{\partial t}=-\nabla\times\bfE\thinspace,\\
\frac{\partial\bfE}{\partial t}=\nabla\times\bfB-\sum_{s}\int\rmd\bfv q_{s}f_{s}\left(\bfx,\bfv,t\right)\bfv\thinspace,\label{eq:AL}
\end{alignat}
where $\bfE$ and $\bfB$ are electromagnetic fields, $f_{s}$, $m_{s}$
and $q_{s}$ are the number density distribution function, mass and
charge of the $s$'th particle, respectively. Here the permittivity
and permeability are set to 1 for simple notation. This set of equations
is also a Hamiltonian Partial Differential Equation (PDE) \cite{morrison1980maxwell,weinstein1981comments,marsden1982hamiltonian},
which means that solutions of the equations conserve the symplectic
structure \cite{marsden1998multisymplectic,marsden2013introduction}
and various types of invariants. As one of these invariants, the Hamiltonian
itself is conserved, i.e.,
\begin{eqnarray}
\frac{\partial}{\partial t}\int\rmd\bfx\left(\frac{E^{2}+B^{2}}{2}+\sum_{s}\int\rmd\bfv\left(\frac{1}{2}m_{s}v^{2}f_{s}\left(\bfx,\bfv;t\right)\right)\right)=0~.
\end{eqnarray}
For the Vlasov-Maxwell equations and many other Hamiltonian PDE systems,
there are also local conservation laws, which can be written in the
following form 
\begin{eqnarray}
\frac{\partial p}{\partial t}+\nabla\cdot\bfu=0~,
\end{eqnarray}
where $p$ is a scalar field or a component of a tensor or vector
field, and $\bfu$ is the flux corresponding to $p$. As an important
example, the local energy conservation law for the Vlasov-Maxwell
system reads 
\begin{eqnarray}
 &  & \frac{\partial}{\partial t}\left(\frac{E^{2}+B^{2}}{2}+\sum_{s}\int\rmd\bfv\left(\frac{1}{2}m_{s}v^{2}f_{s}\left(\bfx,\bfv;t\right)\right)\right)+\nonumber \\
 &  & \nabla\cdot\left(\bfE\times\bfB+\sum_{s}\int\rmd\bfv\left(\frac{1}{2}m_{s}v^{2}\bfv f_{s}\left(\bfx,\bfv;t\right)\right)\right)=0~.\label{EqnENEcons}
\end{eqnarray}
It can be verified directly by using the Vlasov-Maxwell equations
\eqref{eq:VM}-\eqref{eq:AL}. The conservation law can also be obtained
by using Noether's theorem and the weak Euler-Lagrangian equations
\cite{qin2014field}. Local conservation laws are more fundamental
and practical than global conservation laws, because we often consider
systems without global conservations, for example, systems with open
boundaries and particle sources. Nowadays, Particle-In-Cell (PIC)
simulations \cite{birdsall1991plasma,hockney1988computer} are commonly
used in the investigation of Vlasov-Maxwell systems \cite{Cary93,nieter2004vorpal,Xiang08,germaschewski2016plasma,huang16,Qiang2016,Oeftiger16,Planche16},
and advanced structure-preserving geometric algorithms based on variational
or Hamiltonian discretization have been developed recently \cite{Squire4748,squire2012geometric,xiao2013variational,kraus2013variational,evstatiev2013variational,Shadwick14,xiao2015variational,xiao2015explicit,crouseilles2015hamiltonian,Qin15JCP,he2015hamiltonian,he2016hamiltonian,qin2016canonical,Webb16,kraus2016gempic,xiao2016explicit}.
Some of these structure-preserving PIC schemes are able to bound the
global energy errors for all simulation time-steps and are effective
for solving multi-scale problems \cite{xiao2013variational,xiao2015explicit,xiao2016explicit}.
For these algorithms, the corresponding spatially-discretized systems
conserve global energy as a consequence of time symmetry admitted
by the system. However, as an essential physical property, the discrete
local energy conservation has not yet been discussed. If a discrete
system has a local conservation law, it implies that at every grid
point there is a conservation law, which is much stronger than just
one global conservation law for the entire system. There have been
some investigations on the discrete local energy conservation law
associated with Yee's FDTD schemes for Maxwell's equations \cite{chew1994electromagnetic,de1995poynting}.
However, the local conservation law for discrete particle-field systems
is still an unexplored topic. In this paper, we prove a discrete local
energy conservation law for the spatially-discretized Vlasov-Maxwell
system described in Ref. \cite{xiao2015explicit}. This is accomplished
by proving in the Appendix a theorem stating that for a general spatially-discretized
system, a global conservation law always implies a discrete local
conservation law in space-time when the algorithm is local. Here,
an algorithm is called local if the time-advance of a field at a grid
point or a particle involving only its neighboring grid points and
particles. With this theorem, for any geometric spatial-discretizations
with local algorithms, we only need to search for global conservation
laws and then the corresponding local conservation laws are automatically
satisfied. This general result demonstrates that Hamiltonian discretizations
can preserve local conservation laws in space-time, in addition to
the symplectic structure, both of which are the intrinsic physical
properties of many important infinite dimensional Hamiltonian systems
in physics. 

The idea of structure-preserving spatial discretization of the Vlasov-Maxwell
equations can be traced back to Lewis \cite{lewis1970energy} who
proposed a spatial-discretized Lagrangian for Vlasov plasmas. Today,
modern geometric discretizations for constructing PIC schemes are
based on Discrete Exterior Calculus (DEC) \cite{hirani2003discrete,desbrun2005discrete},
interpolation forms \cite{whitney1957geometric,desbrun2008discrete,xiao2015explicit}
or Finite Element Exterior Calculus (FEEC) \cite{arnold2006finite,arnold2010finite,monk2003finite}
to ensure the conservation of charge, the gauge invariance, and the
symplectic structure \cite{Squire4748,squire2012geometric,xiao2015explicit,he2016hamiltonian,kraus2016gempic,xiao2016explicit}.
We start from the Lagrangian of the spatially-discretized Vlasov-Maxwell
system in Ref. \cite{xiao2015explicit}, 
\begin{multline}
L_{sd}=\frac{1}{2}\left(\sum_{J}\left(-\dot{\bfA}_{J}-\sum_{I}{\nabla_{\mathrm{d}}}_{JI}\phi_{I}\right)^{2}-\sum_{K}\left(\sum_{J}\CURLD_{KJ}\bfA_{J}\right)^{2}\right)\Delta V+\\
\sum_{s}\left(\frac{1}{2}m_{s}\dot{\bfx}_{s}^{2}+q_{s}\left(\dot{\bfx}_{s}\cdot\sum_{J}W_{\sigma_{1J}}\left(\bfx_{s}\right)\bfA_{J}-\sum_{I}W_{\sigma_{0I}}\left(\bfx_{s}\right)\phi_{I}\right)\right)~.
\end{multline}
Here, integers $I$, $J$ and $K$ are indices of grid points in a
cubic-mesh, and $\nabla_{\mathrm{d}}$, $\CURLD$ and $\DIVD$ are
the discrete gradient, curl and divergence operators defined as follows,
\begin{eqnarray}
\left({\nabla_{\mathrm{d}}}\phi\right)_{i,j,k} & = & [\phi_{i+1,j,k}-\phi_{i,j,k},\phi_{i,j+1,k}-\phi_{i,j,k},\phi_{i,j,k+1}-\phi_{i,j,k}]~,\label{EqnDEFGRADD}\\
\left(\CURLD\bfA\right)_{i,j,k} & = & \left[\begin{array}{c}
\left({A_{z}}_{i,j+1,k}-{A_{z}}_{i,j,k}\right)-\left({A_{y}}_{i,j,k+1}-{A_{y}}_{i,j,k}\right)\\
\left({A_{x}}_{i,j,k+1}-{A_{x}}_{i,j,k}\right)-\left({A_{z}}_{i+1,j,k}-{A_{z}}_{i,j,k}\right)\\
\left({A_{y}}_{i+1,j,k}-{A_{y}}_{i,j,k}\right)-\left({A_{x}}_{i,j+1,k}-{A_{x}}_{i,j,k}\right)
\end{array}\right]^{T}~,\label{DEFCURLD}\\
\left({\DIVD}\bfB\right)_{i,j,k} & = & \left({B_{x}}_{i+1,j,k}-{B_{x}}_{i,j,k}\right)+\left({B_{y}}_{i,j+1,k}-{B_{y}}_{i,j,k}\right)+\nonumber \\
 &  & \left({B_{z}}_{i,j,k+1}-{B_{z}}_{i,j,k}\right)~.\label{EqnDEFDIVD}
\end{eqnarray}
These operators are local linear operators on the discrete fields
$\phi_{I}$, $\bfA_{J}$ and $\bfB_{K}$. Functions $W_{\sigma_{0J}}$,
$W_{\sigma_{1I}}$ and $W_{\sigma_{2K}}$ are interpolation functions
(Whitney forms) for 0-forms, e.g., scalar potential, 1-forms, e.g.,
vector potential, and 2-forms, e.g., magnetic fields, respectively.
These discrete operators and interpolation functions satisfy the following
properties \cite{whitney1957geometric,hirani2003discrete,desbrun2008discrete,xiao2015explicit},
\begin{eqnarray}
\nabla\sum_{I}W_{\sigma_{0I}}\left(\bfx\right)\phi_{I} & = & \sum_{I,J}W_{\sigma_{1J}}\left(\bfx\right){\nabla_{\mathrm{d}}}_{JI}\phi_{I}~,\label{EqnD0to1FORMAPP}\\
\nabla\times\sum_{J}W_{\sigma_{1J}}\left(\bfx\right)\bfA_{J} & = & \sum_{J,K}W_{\sigma_{2K}}\left(\bfx\right){\CURLD}_{KJ}\bfA_{J}~,\label{EqnD1to2FORMAPP}\\
\nabla\cdot\sum_{K}W_{\sigma_{2K}}\left(\bfx\right)\bfB_{K} & = & \sum_{K,L}W_{\sigma_{3L}}\left(\bfx\right){\DIVD}_{LK}\bfB_{K}~.\label{EqnD2to3FORMAPP}
\end{eqnarray}
Periodic boundary in all three directions are adopted to simplify
the discussion. Because time $t$ does not explicitly appear in the
Lagrangian, the total energy 
\begin{eqnarray}
H_{sd} & = & \frac{\partial L_{sd}}{\partial\dot{q}}\dot{q}^{T}-L_{sd}~\\
 & = & \frac{1}{2}\Delta V\left(\sum_{J}\bfE_{J}^{2}+\sum_{K}\bfB_{K}\right)+\sum_{s}\frac{1}{2}m_{s}\dot{\bfx}_{s}^{2}~
\end{eqnarray}
is conserved, i.e., $\dot{H}_{sd}=0$. Here, $q$ is the generalized
coordinates, i.e., $q=[\bfA_{J},\phi_{I},\bfx_{s}]$, and $\bfE_{J}$
and $\bfB_{K}$ are discrete electromagnetic fields defied as 
\begin{eqnarray}
\bfE_{J} & = & -\dot{\bfA}_{J}-\sum_{J}{\nabla_{\rmd}}_{JI}\phi_{I}~,\\
\bfB_{K} & = & \sum_{J}\CURLD_{KJ}\bfA_{J}~.
\end{eqnarray}
The corresponding Poisson bracket for this system is \cite{xiao2015explicit}
\begin{multline}
\left\{ F,G\right\} =\frac{1}{\Delta V}\sum_{J}\left(\frac{\partial F}{\partial\bfE_{J}}\cdot\sum_{K}\frac{\partial G}{\partial\bfB_{K}}\CURLD_{KJ}-\sum_{K}\frac{\partial F}{\partial\bfB_{K}}\CURLD_{KJ}\cdot\frac{\partial G}{\partial\bfE_{J}}\right)+\\
\sum_{s}\frac{1}{m_{s}}\left(\frac{\partial F}{\partial\bfx_{s}}\cdot\frac{\partial G}{\partial\dot{\bfx}_{s}}-\frac{\partial F}{\partial\dot{\bfx}_{s}}\cdot\frac{\partial G}{\partial\bfx_{s}}\right)+\\
\sum_{s}\frac{q_{s}}{m_{s}\Delta V}\left(\frac{\partial F}{\partial\dot{\bfx}_{s}}\cdot\sum_{J}W_{\sigma_{1J}}\left(\bfx_{s}\right)\frac{\partial G}{\partial\bfE_{J}}-\frac{\partial G}{\partial\dot{\bfx}_{s}}\cdot\sum_{J}W_{\sigma_{1J}}\left(\bfx_{s}\right)\frac{\partial F}{\partial\bfE_{J}}\right)+\\
-\sum_{s}\frac{q_{s}}{m_{s}^{2}}\frac{\partial F}{\partial\dot{\bfx}_{s}}\cdot\left[\sum_{K}W_{\sigma_{2K}}\left(\bfx_{s}\right)\bfB_{K}\right]\times\frac{\partial G}{\partial\dot{\bfx}_{s}}~.\label{eq:22}
\end{multline}
With this Poisson bracket, the time evolution of the system is
\begin{eqnarray}
\dot{g}=\{g,H_{sd}\}~,\label{EqnDEM}
\end{eqnarray}
where $g=[\bfE_{J},\bfB_{K},\bfx_{s},\dot{\bfx}_{s}].$ Now we introduce
the discrete local energy $\epsilon_{I}$ at the $I$'th grid, 
\begin{eqnarray}
\varepsilon_{I}=\sum_{s}\frac{1}{2}m_{s}\dot{\bfx}_{s}^{2}\WZERO{\bfx_{s}}+\bfE_{I}^{2}+\bfB_{I}^{2}.
\end{eqnarray}
The evolution of $\varepsilon_{I}$ is
\begin{eqnarray}
\dot{\varepsilon}_{I} & = & \left\{ \varepsilon_{I},H_{sd}\right\} ~,
\end{eqnarray}
or more specifically, 
\begin{eqnarray}
\dot{\varepsilon}_{I} & = & \sum_{s}q_{s}\left(\dot{\bfx}_{s}\cdot\sum_{J'}\bfE_{J'}\WONEJp{\bfx_{s}}\right)\WZERO{\bfx_{s}}+\sum_{s}\frac{1}{2}m_{s}\dot{\bfx}_{s}^{2}\left(\dot{\bfx}_{s}\cdot\nabla\WZERO{\bfx_{s}}\right)+\nonumber \\
 &  & \bfE_{I}\cdot\sum_{K}\CURLD_{KI}\bfB_{K}-\bfE_{I}\cdot\sum_{s}q_{s}\dot{\bfx}_{s}\WONEA{\bfx_{s}}-\bfB_{I}\cdot\sum_{J}\CURLD_{IJ}\bfE_{J}~.\label{EqnDEPSLDT}
\end{eqnarray}
We will see that the right hand side of \EQ{EqnDEPSLDT} can be
written as a discrete divergence of a discrete vector field, which
means that \EQ{EqnDEPSLDT} is a discrete energy conservation law.

Let us divide the RHS of \EQ{EqnDEPSLDT} into three terms,
\begin{eqnarray}
\dot{\varepsilon}_{I} & = & T_{1}+T_{2}+T_{3}~,
\end{eqnarray}
where 
\begin{eqnarray}
T_{1} & = & \bfE_{I}\cdot\sum_{K}\CURLD_{KI}\bfB_{K}-\bfB_{I}\cdot\sum_{J}\CURLD_{IJ}\bfE_{J}~,\\
T_{2} & = & \sum_{s}\frac{1}{2}m_{s}\dot{\bfx}_{s}^{2}\left(\dot{\bfx}_{s}\cdot\nabla\WZERO{\bfx_{s}}\right)~,\\
T_{3} & = & \sum_{s}q_{s}\left(\dot{\bfx}_{s}\cdot\sum_{J'}\bfE_{J'}\WONEJp{\bfx_{s}}\right)\WZERO{\bfx_{s}}-\bfE_{I}\cdot\sum_{s}q_{s}\dot{\bfx}_{s}\WONEA{\bfx_{s}}~.
\end{eqnarray}
Firstly, for $T_{1}$, we can check that this term can be written
as a discrete divergence 
\begin{eqnarray}
T_{1}=-\sum_{K}\DIVD_{IK}(\bfE\times*\bfB)_{K}~,
\end{eqnarray}
where $\bfE\times*\bfB$ is defined as 
\begin{eqnarray}
(\bfE\times*\bfB)_{i,j,k} & = & \left[\begin{array}{c}
{E_{y}}_{i,j,k}{B_{z}}_{i-1,j,k}-{E_{z}}_{i,j,k}{B_{y}}_{i-1,j,k}\\
{E_{z}}_{i,j,k}{B_{x}}_{i,j-1,k}-{E_{x}}_{i,j,k}{B_{z}}_{i,j-1,k}\\
{E_{x}}_{i,j,k}{B_{y}}_{i,j,k-1}-{E_{y}}_{i,j,k}{B_{x}}_{i,j,k-1}
\end{array}\right]^{T}.
\end{eqnarray}
Next, we investigate $T_{2}$ which represents the energy flow of
particles, 
\begin{eqnarray}
T_{2} & = & \sum_{s}\frac{1}{2}m_{s}\dot{\bfx}_{s}^{2}\left(\dot{\bfx}_{s}\cdot\nabla\WZERO{\bfx_{s}}\right)\nonumber \\
 & = & \sum_{s}\frac{1}{2}m_{s}\dot{\bfx}_{s}^{2}\left(\dot{\bfx}_{s}\cdot\sum_{J}\WONE{\bfx_{s}}{\nabla_{d}}_{JI}\right)\nonumber \\
 & = & \sum_{J}{\nabla_{d}}_{JI}\sum_{s}\frac{1}{2}m_{s}\dot{\bfx}_{s}^{2}\dot{\bfx}_{s}\cdot\WONE{\bfx_{s}}\nonumber \\
 & = & \sum_{J}{\nabla_{d}}_{JI}\mathbf{S}_{J}=-\sum_{K}\DIVD_{IK}\left(*\mathbf{S}\right)_{K}~,\label{eq:T2}
\end{eqnarray}
where 
\begin{eqnarray}
\mathbf{S}_{J} & = & \sum_{s}\frac{1}{2}m_{s}\dot{\bfx}_{s}^{2}\dot{\bfx}_{s}\cdot\WONE{\bfx_{s}}~,\\
\left(*\mathbf{S}\right)_{K} & = & [{S_{x}}_{i-1,j,k},{S_{y}}_{i,j-1,k},{S_{z}}_{i,j,k-1}]~.
\end{eqnarray}
Thus this term is also a discrete divergence. Finally, let us look
at $T_{3}$, which appears only in the discrete particle-field system.
It is 
\begin{eqnarray}
T_{3} & = & \sum_{s}q_{s}\left(\dot{\bfx}_{s}\cdot\sum_{J}\bfE_{J}\WONE{\bfx_{s}}\right)\WZERO{\bfx_{s}}-\bfE_{I}\cdot\sum_{s}q_{s}\dot{\bfx}_{s}\WONEA{\bfx_{s}}\\
 & = & \sum_{s}q_{s}\left(\sum_{J'}\dot{\bfx}_{s}\cdot\bfE_{J}\WONE{\bfx_{s}}\WZERO{\bfx_{s}}-\bfE_{I}\cdot\dot{\bfx}_{s}\WONEA{\bfx_{s}}\right)~.\\
 & = & \sum_{s}q_{s}F\left(s\right)\left(\WZERO{\bfx_{s}}-\frac{p\left(s,I\right)}{F\left(s\right)}\right)~,
\end{eqnarray}
where
\begin{eqnarray}
F\left(s\right) & = & \sum_{I}p\left(s,I\right)\thinspace,\\
p\left(s,I\right) & = & \bfE_{I}\cdot\dot{\bfx}_{s}\WONEA{\bfx_{s}}~.
\end{eqnarray}
From the definition of $\WZERO{\bfx}$, we have $\sum_{I}\WZERO{\bfx_{s}}=1~.$
Therefore,
\begin{eqnarray}
\sum_{I}\left(\WZERO{\bfx_{s}}-\frac{p\left(s,I\right)}{F\left(s\right)}\right) & = & 0~,
\end{eqnarray}
which means that the sum of $T_{3}$ over all spatial grid points
vanishes. Now we invoke the two theorems approved in the Appendix,
which states that a discrete sum-free field must be a discrete divergence
field and that if the sum-free field is local, then the divergence
field is local. Therefore, for each particle, $T_{3}$ can be expressed
as discrete divergence of a discrete vector field $\mathbf{G}_{s}$,
\begin{eqnarray}
T_{3} & = & \sum_{s}q_{s}\sum_{J'}\DIVD_{IJ'}{\mathbf{G}_{s}}_{J'}F\left(s\right)~.
\end{eqnarray}
 For a particular $s$, $\WZERO{\bfx_{s}}$, $\WONE{\bfx_{s}}$ are
local near $\bfx_{s}$, so is $\WZERO{\bfx_{s}}-p\left(s,I\right)/F(s)$
. Thus $\bfG_{s}$ is also local near $\bfx_{s}$. 

Finally, we obtain the local energy conservation law for the spatial-discretized
Vlasov-Maxwell system. 
\begin{eqnarray}
 &  & \frac{\partial}{\partial t}\left(\sum_{s}\frac{1}{2}m_{s}\dot{\bfx}_{s}^{2}\WZERO{\bfx_{s}}+\bfE_{I}^{2}+\bfB_{I}^{2}\right)+\nonumber \\
 &  & \sum_{K}\DIVD_{IK}\left((\bfE\times*\bfB)_{K}+\left(*\mathbf{S}\right)_{K}-\sum_{s}q_{s}{\mathbf{G}_{s}}_{K}F\left(s\right)\right)=0~.
\end{eqnarray}

In conclusion, we started from the Hamiltonian theory of the spatially-discretized
Vlasov-Maxwell system, used the property of Whitney forms, and derived
a discrete local energy conservation law. In practice the spatially-discretized
Vlasov-Maxwell system also need a temporal discretize to become a
numerical scheme, after which the total energy will not be an exact
invariant. However if we apply a symplectic integrator to perform
the temporal discretize, then the total energy error can be bounded
within a small value for all simulation time-steps \cite{Ruth83,Feng85,Feng86,Hairer02}.
Investigation on discrete conservation laws for such systems are planned
for future work.

\appendix

\section{a discrete sum-free field is a discrete divergence field}

In this appendix, we prove the following two theorems.

\newtheorem{theorem1}{Theorem} \newtheorem{lemma}{Lemma} \begin{theorem1}\label{LemR}
If a discrete scalar field $R$ is sum-free, i.e.,
\begin{eqnarray}
\sum_{I}R_{I} & = & 0~,
\end{eqnarray}
then it can be expressed as a discrete divergence of a discrete vector
field $\mathbf{G}$, i.e., 
\begin{eqnarray}
R_{I} & = & \sum_{J}\DIVD_{IJ}\mathbf{G}_{J}~.\label{EqnRIGJ}
\end{eqnarray}
\end{theorem1}

The discrete field $R$ in Theorem \ref{LemR} can also be a function
of the continuous spatial coordinate $\bfx$, such as the Whitney
form $W_{\sigma_{0J}}(\bfx).$ A discrete field or a field component
$F(\bfx)$ is called local, if for every $\bfx_{0}$, there exists
a positive constant $C$ such that 
\begin{eqnarray}
F_{I}(\bfx_{0})=0\thinspace,\forall I\in\left\{ J|\quad|\bfx_{J}-\bfx_{0}|>C\right\} .
\end{eqnarray}

\begin{theorem1}\label{LemR2} If a discrete sum-free scalar field
$R$ is local, then there exists a local discrete vector field $\bfG$
such that Eq.\,\eqref{EqnRIGJ} holds. \end{theorem1}

These two theorems also imply that a (local) discrete sum-free vector
field or tensor field can be expressed as a discrete divergence of
a (local) tensor or high order tensor, because for each component
of the vector field or tensor field, Theorems \ref{LemR} and \ref{LemR2}
apply. To prove these two theorems, we first prove the following three
lemmas. \begin{lemma} The discrete scalar field $R_{x}(i',j',k')$
defied as 
\begin{eqnarray}
R_{xI}(i',j',k')=\left\{ \begin{array}{lc}
1, & \textrm{if }I=[i'+1,j',k']~,\\
-1, & \textrm{if }I=[i',j',k']~,\\
0, & \textrm{otherwise}~,
\end{array}\right.
\end{eqnarray}
can be expressed as a composition of $\DIVD$ and a discrete vector
field $\mathbf{G}_{x}(i',j',k')$, 
\begin{eqnarray}
R_{xI}(i',j',k') & = & \sum_{J}\DIVD_{IJ}\mathbf{G}_{xJ}(i',j',k')~.
\end{eqnarray}
\end{lemma} \begin{proof} Let $\bfG_{x}\left(i',j',k'\right)$ be
\begin{eqnarray}
\bfG_{xJ}\left(i',j',k'\right)=\left\{ \begin{array}{lc}
[-1,0,0], & \textrm{if }J=[i'+1,j',k']~,\\
\left[0,0,0\right], & \textrm{otherwise}~.
\end{array}\right.
\end{eqnarray}
Then it is straightforward to verify that 
\begin{eqnarray}
R_{xI}(i',j',k') & = & \sum_{J}\DIVD_{IJ}\mathbf{G}_{xJ}(i',j',k')~.
\end{eqnarray}
\end{proof} Using the similar technique, we can construct $\bfG_{y}\left(i',j',k'\right)$
and $\bfG_{z}\left(i',j',k'\right)$ for scalar fields $R_{y}(i',j',k')$
and $R_{z}(i',j',k')$ as well, i.e.,
\begin{eqnarray}
R_{yI}(i',j',k') & = & \left\{ \begin{array}{lc}
1, & \textrm{if }I=[i',j'+1,k']~,\\
-1, & \textrm{if }I=[i',j',k']~,\\
0, & \textrm{otherwise}~,
\end{array}\right.\\
R_{zI}(i',j',k') & = & \left\{ \begin{array}{lc}
1, & \textrm{if }I=[i',j',k'+1]~,\\
-1, & \textrm{if }I=[i',j',k']~,\\
0, & \textrm{otherwise}~,
\end{array}\right.\\
\bfG_{yJ}\left(i',j',k'\right) & = & \left\{ \begin{array}{lc}
[0,-1,0], & \textrm{if }J=[i',j'+1,k']~,\\
\left[0,0,0\right], & \textrm{otherwise}~,
\end{array}\right.\\
\bfG_{zJ}\left(i',j',k'\right) & = & \left\{ \begin{array}{lc}
[0,0,-1], & \textrm{if }J=[i',j',k'+1]~,\\
\left[0,0,0\right], & \textrm{otherwise}~.
\end{array}\right.
\end{eqnarray}

\begin{lemma} The discrete scalar field $R^{\dagger}\left(I_{1},I_{2}\right)$
defined as 
\begin{eqnarray}
R_{I}^{\dagger}\left(I_{1},I_{2}\right) & = & \left\{ \begin{array}{lc}
-1, & \textrm{if }I=I_{1}~,\\
1, & \textrm{if }I=I_{2}~,\\
0, & \textrm{otherwise}
\end{array}\right.
\end{eqnarray}
can be written in the following form 
\begin{eqnarray}
R_{I}^{\dagger}\left(I_{1},I_{2}\right) & = & \sum_{I'\in Y_{x}}a_{x}R_{xI}(I')+\sum_{I'\in Y_{y}}a_{y}R_{yI}(I')+\sum_{I'\in Y_{z}}a_{z}R_{zI}(I')~,\label{EqnRDAGGER}
\end{eqnarray}
where $Y_{x}$, $Y_{y}$ and $Y_{z}$ are some indices sets, and $a_{x}$,
$a_{y}$ and $a_{z}$ are some integers. \end{lemma} \begin{proof}
Choose $Y_{x}$, $Y_{y}$ and $Y_{z}$ as 
\begin{eqnarray}
Y_{x} & = & \left\{ [i,j,k]|\min(i_{1},i_{2})\leq i<\max(i_{1},i_{2}),j=j_{2},k=k_{2}\right\} ~,\\
Y_{y} & = & \left\{ [i,j,k]|i=i_{1},\min(j_{1},j_{2})\leq j<\max(j_{1},j_{2}),k=k_{2}\right\} ~,\\
Y_{z} & = & \left\{ [i,j,k]|i=i_{1},j=j_{1},\min(k_{1},k_{2})\leq k<\max(k_{1},k_{2})\right\} ~.
\end{eqnarray}
Choose $a_{x}$, $a_{y}$ and $a_{z}$ as 
\begin{eqnarray}
a_{x} & = & \left\{ \begin{array}{lc}
1, & \textrm{if }i_{1}<i_{2}~,\\
-1, & \textrm{otherwise}~,
\end{array}\right.\\
a_{y} & = & \left\{ \begin{array}{lc}
1, & \textrm{if }j_{1}<j_{2}~,\\
-1, & \textrm{otherwise}~,
\end{array}\right.\\
a_{z} & = & \left\{ \begin{array}{lc}
1, & \textrm{if }k_{1}<k_{2}~,\\
-1, & \textrm{otherwise}~.
\end{array}\right.
\end{eqnarray}
It is straightforward to verify that \EQ{EqnRDAGGER} holds. \end{proof} 

\begin{lemma} A sum-free scalar field $R$ can be written as 
\begin{eqnarray}
R_{I} & = & \sum_{[I',I'']\in Z}b_{[I',I'']}R_{I}^{\dagger}\left(I',I''\right)~,\label{EqnSUMRB}
\end{eqnarray}
where $Z$ is some index-pair set. \end{lemma}\begin{proof} If $\forall I$
such that $R_{I}=0$, then $Z$ can be chosen as $\varnothing$. Otherwise
let $Y_{d}=\left\{ I_{1},I_{2},I_{3},\dots\right\} $ is the index
set such that if $I\in Y_{d}$, $R_{I}\neq0$ and $I\notin Y_{d}$,
$R_{I}=0$. We can choose the index-pair set $Z$ as 
\begin{eqnarray}
Z=\left\{ [I',I'']|I'=I_{1},I''\in Y_{d}\textrm{ and }I''\neq I_{1}\right\} ,
\end{eqnarray}
where $I_{1}$ is one arbitrarily chosen element in $Y_{d}$, and
the corresponding $b_{[I',I'']}$ is 
\begin{eqnarray}
b_{[I',I'']}=R_{I''}.
\end{eqnarray}
Using the fact that 

\begin{equation}
R_{I_{1}}+\sum_{I''\in Y_{d}\textrm{ and }I''\neq I_{1}}R_{I''}=\sum_{I}R_{I}=0~,
\end{equation}
we have

\begin{equation}
R_{I_{1}}=-\sum_{I''\in Y_{d}\textrm{ and }I''\neq I_{1}}R_{I''}~.
\end{equation}
Then it is straightforward to verify that \EQ{EqnSUMRB} holds.
\end{proof} Composing these three lemmas, we can see that the Theorems
\ref{LemR} and \ref{LemR2} are proved. 
\begin{acknowledgments}
This research is supported by National Magnetic Confinement Fusion
Energy Research Project (2015GB111003, 2014GB124005), National Natural
Science Foundation of China (NSFC-11575185, 11575186, 11305171), JSPS-NRF-NSFC
A3 Foresight Program (NSFC-11261140328), Chinese Scholar Council (201506340103),
Key Research Program of Frontier Sciences CAS (QYZDB-SSW-SYS004),
and the GeoAlgorithmic Plasma Simulator (GAPS) Project.
\end{acknowledgments}

 \bibliographystyle{apsrev4-1}
\bibliography{enecons_pic}

\end{document}